\documentclass[12pt,a4paper]{article}

\title{On the Design of Paired Comparison Experiments with Application}
\author{Eric Nyarko\footnote{corresponding author.~E-mail:~\texttt{ericnyarko@ug.edu.gh, nyarkoeric5@gmail.com}}\\
University of Ghana, \\ Department of Statistics and Actuarial Science, \\ Box LG 115, Legon-Accra}
\usepackage{booktabs}
\usepackage{dirtytalk}
\usepackage{float,rotating}
\usepackage{diagbox}
\usepackage{slashbox}
\usepackage{appendix}
\usepackage{tikz}
\usepackage{xcolor}
\usepackage{url}
\usepackage[square,numbers]{natbib}
\setcitestyle{authoryear,open={(},close={)}}
\usepackage{array}
\usepackage{longtable,pdflscape,caption}
\usepackage{amsmath,mathtools}
\usepackage{amsthm}
\usepackage{siunitx} 
\usepackage{booktabs}       
\usepackage{dcolumn}        
\usepackage[flushleft]{threeparttable}

\usepackage{dcolumn}
\usepackage{lipsum} 

\date{}
\begin{document}
\maketitle
\begin{abstract}\noindent
In practice, paired comparison experiments involving pairs of either full or partial profiles are frequently used. When all attributes have a general common number of levels, the problem of finding optimal designs is considered in the presence of a second-order interactions model. The $D $-optimal designs for the second-order interactions model in this setting have both types of pairs in which either all attributes have different levels or approximately half of the attributes are different. The proposed optimal designs can be used as a benchmark to compare any design for estimating main effects and two and three attribute interactions. A practical situation that incorporates the corresponding second-order interactions is covered.
\end{abstract}
{\bf Keywords:}~Full profile; Interactions; Optimal design; Paired comparisons; Partial profile; Profile strength\\\vspace{0.5mm}

\noindent
{\it AMS 2000 Subject Classifications}:~Primary:~62K05;~Secondary:~62J15,~62K15
\section{Introduction}
Paired comparisons are similar to experiments with two options (alternatives), which are widely used in many fields of application such as health economics, transportation economics, education and marketing to study people's or consumers' preferences toward new products or services where behaviors of interest typically involve quantitative or qualitative responses \citep[see e.g.,] [] {scheffe1952analysis, green1990conjoint, grasshoff2003optimal}. The present paper draws on this situation of paired comparisons (so called conjoint analysis where responses are usually assessed on a rating scale or discrete choice experiment) as frequently encountered in practice. \par

Typically, in paired comparison experiments, respondents trade-off one alternative against the other (in a hypothetical or occasional real setting) which is generated by an experimental design, and are described by a number of attributes. Usually, in applications one may be interested in main effects as well as interactions between the attributes. For example, \citet{bradley1976treatment} considered an analysis of a factorial experiment on coffee preferences, where main effects plus two-attribute interactions plus three-attribute interactions were of special interest \citep[see also][]{el1978treatment, louviere2000stated}. An experiment in a food manufacturing research where interactions between three ingredients: a milk product, a starch product and an emulsifier were of interest can be found \citep{lewis1985paired}. A study that investigates the extent to which the policy of prioritizing life-extending end-of-life treatments over other types of treatment is consistent with the stated preferences of members of the general public, and which incorporates the corresponding situation of second-order (or three-attribute) interactions can be found \citep{shah2015valuing}. Moreover, \citet{elrod1992empirical} considered main effects and up to four-attribute interactions for a study on student's preferences for rental apartments. The present paper where any of three of the attributes interact is motivated by the aforementioned works. \par

Due to the limited cognitive ability to process information in applications, a paired comparison task including many attributes may result in respondent decisions that do not reflect their actual preferences. A way to overcome these behaviors is to simplify the paired comparison task by holding the levels of some of the attributes constant in every pair.  The profiles in such  pair are called partial profiles, and the number of attributes that are allowed with potentially different levels in the partial profiles is called the profile strength
\citep[e.g., see][]{green1974design,grasshoff2003optimal,chrzan2010using,kessels2011bayesian}.  \par

In this paper we introduce an appropriate model for the situation of full and partial profiles and derive optimal designs in the presence of second-order (or three-attribute) interactions. We consider the case when the alternatives are specified by general common number of level-attributes. A practical situation of interest that incorporates the corresponding second-order interactions is also considered. In the statistical literature, work on determining the structure of the optimal designs by the two-level situation when the first- and last-levels of the attributes were effects-coded as $1$ and $-1$, respectively has been investigated by \citet{van1987optimal} in the case of full profiles and by \citet{schwabe2003optimal} in the case of partial profiles in a main-effects and two-attribute interactions setup. Corresponding results when the common number of the attribute levels is larger than two have been obtained \citep[see][]{grasshoff2003optimal} in a first-order interactions setup for both full and partial profiles. Here, we treat the case of three-attribute interactions when the common number of the attribute levels is larger than two and provide a detailed proof \citep[see][arXiv]{nyarko2019optimalvlevels}. The two-level situation has been investigated by \citet{nyarko2019optimal} in the case of both full and partial profiles.\par
 It is worthwhile mentioning that the invariant designs considered in this paper possess large number of comparisons. However, these designs can serve as a benchmark to judge the efficiency of competing designs as well as a starting point to construct exact designs or fractions which share the property of optimality and can be realized with a reasonable number of comparisons.\par

The remainder of the paper is organized as follows. A general model is introduced in Section 2 for linear paired comparisons which is related to experiments with two options. This is followed by a three-attribute interactions model for full and partial profiles in Section 3 and optimal designs are characterized in Section 4. A practical situation of importance is considered in Section 5 and the final Section 6 offers some conclusions. 
\section{General setting} 
In any experimental situation the outcome of the experiment depends on some factors (attributes), say, $K$ of influence. In this setting the dependence can be described by a vector of regression functions $\textbf{f}$. In what follows we define a single alternative by $\textbf{i}= (i_{1},\dots,i_{K})$ where $i_k$ is the component of the $k$th attribute, $k=1,\dots,K$. Any utility (not observe) $\tilde{Y}_{n}(\textbf{i})$ of the single alternative $\textbf{i}= (i_{1},\dots,i_{K})$ subject to a block effect $\mu_n$ and a random error $\tilde{\varepsilon}_{n}$, which is assumed to be uncorrelated with constant variance and zero mean can be formalized by a general linear model
\begin{equation}\label{eq:1}
\begin{split}
\tilde{Y}_{n}(\textbf{i})&= \mu_{n}+\textbf{f}(\textbf{i})^{\top} \boldsymbol{\beta} + \tilde{\varepsilon}_{n},
\end{split}
\end{equation}
where the index $n$ denotes the $n$th presentation, $n=1,\dots,N$, and the alternative $\textbf{i}$ is chosen from a set $\mathcal{I}=\{1,\dots,v\}^K$. Here the vector of known regression functions $\textbf{f}$ describe the form of the functional relationship between the alternative $\textbf{i}$ and the corresponding mean response $E(\tilde{Y}_{n}(\textbf{i}))=\mu_n+\textbf{f}(\textbf{i})^{\top}\boldsymbol{\beta}$, and $\boldsymbol{\beta}$ is the unknown parameter vector of interest. Usually in order to make statistical inference on the unknown parameters several pairs are presented to get rid of the influence of the block effect $\mu_{n}$ due to a variety of unobservable influences. Then actual differences of the latent utilities are observed for the alternatives presented in a pair. \par

More specifically, unlike in standard experimental designs where there is a possibility of only a single or direct observation, in paired comparison experiments the utilities for the alternatives are usually not directly observed. Only observations $Y_n(\textbf{i},\textbf{j})=\tilde{Y}_{n}(\textbf{i})-\tilde{Y}_{n}(\textbf{j})$ are available for comparing pairs $(\textbf{i},\textbf{j})$ of alternatives $\textbf{i}$ and $\textbf{j}$ which are chosen from the design region $\mathcal{X}=\mathcal{I}\times \mathcal{I}$. In that case the utilities for the alternatives are properly described by the linear paired comparison model 
\begin{equation} \label{eq:2}
\begin{split}
Y_n(\textbf{i, j})=(\textbf{f}(\textbf{i})-\textbf{f}(\textbf{j}))^{\top}\boldsymbol\beta+\varepsilon_{n}, \\
\end{split}
\end{equation}
where $\textbf{f}(\textbf{i})-\textbf{f}(\textbf{j})$ is the derived regression function and the random errors $\varepsilon_{n}(\textbf{i}, \textbf{j})=\tilde{\varepsilon}_{n}(\textbf{i})-\tilde{\varepsilon}_{n}(\textbf{j})$ associated with the different pairs $(\textbf{i}, \textbf{j})$ are assumed to be uncorrelated with constant variance and zero mean. Here, the block effects $\mu_{n}$ are immaterial. \par

The performance of the statistical analysis based on a paired comparison experiment depends on the pairs in the preference task that are presented. The choice of such pairs $(\textbf{i}_1,\textbf{j}_1),\dots,(\textbf{i}_N,\textbf{j}_N)$ is called a design of size $N$. 
The quality of such a design is measured by its information matrix
\begin{equation}\label{eq:3}
\textbf{M}((\textbf{i}_1, \textbf{j}_1),\dots,(\textbf{i}_N, \textbf{j}_N))=\sum_{n=1}^{N}\textbf{M}((\textbf{i}_n, \textbf{j}_n)),
\end{equation}
where $\textbf{M}((\textbf{i},\textbf{j}))=(\textbf{f}(\textbf{i})-\textbf{f}(\textbf{j}))(\textbf{f}(\textbf{i})-\textbf{f}(\textbf{j}))^{\top}$ is the so-called elemental information of a single pair $(\textbf{i},\textbf{j})$.\par

In this paper we restrict our attention to approximate or continuous designs $\xi$ \citep[e.g., see][]{kiefer1959optimum} which are defined as discrete probability measures on the design region $\mathcal{X}$ of all pairs $(\textbf{i},\textbf{j})$. Moreover, every approximate design $\xi$ which assigns only rational weigths $\xi(\textbf{i},\textbf{j})$ to all pairs $(\textbf{i},\textbf{j})$ in its support points can be realized as an exact design $\xi_N$ of size $N$ consisting of the pairs $(\textbf{i}_1,\textbf{j}_1),\dots,(\textbf{i}_N,\textbf{j}_N)$. Note that for an exact design $\xi_N=((\textbf{i}_1,\textbf{j}_1),\dots,(\textbf{i}_N,\textbf{j}_N))$ the normalized information matrix $\textbf{M}(\xi_N)$ coincides with the information matrix $\textbf{M}(\xi)$ of the corresponding approximate design $\xi$.\par

Optimality criteria for approximate designs $\xi$ are functionals of $\textbf{M}(\xi)$. As a scalar measure of design quality here we consider the criterion of $D$-optimality. An approximate design $\xi^{\ast}$ is $D$-optimal if it maximizes the determinant of the information matrix, that is, if $\mathrm{det} \textbf{M}(\xi^{\ast})$ $\geq$ $\mathrm{det} \textbf{M}(\xi)$ for every approximate design $\xi$ on $\mathcal{X}$.
\vspace{3mm}\\
\noindent
\textbf{Example 1. } One-way layout (general levels)\par\noindent
For illustrative purposes we first consider the situation of just a single-attribute ($K=1$) \citep[e.g., see][]{grasshoff2004optimal} which may vary only over $v$ levels, and adopt the standard parameterization of effects-coding. 
In this setting, the effects of each single level $i=1,\dots,v$ has parameters $\beta_{i}$. Hence,
\begin{equation}\label{eq:5}
\begin{split}
\tilde{Y}_{n}(i)&=\mu_{n}+ \beta_{i}+\tilde{\epsilon}_{n},
\end{split}
\end{equation}
with $i\in\mathcal{I}=\{1,\dots,v\}$, where $\mu_{n}$ denotes the block effect, $n=1,\dots,N$, and $\tilde{\epsilon}_{n}$ the random error is assumed to be uncorrelated with constant variance and zero mean. For effects-coding the regression function $\textbf{f}$ is given by $\textbf{f}(i)=\textbf{e}_i$, $i=1,\dots,v-1$ and $\textbf{f}(v)=-\mathbf{1}_{v-1}$, respectively, where $\textbf{e}_i$ is the $i$th unit vector of length $v-1$ and $\mathbf{1}_{m}$ denotes a vector of length $m$ with all entries equal to $1$. This parameterization is an indication that suitable identifiability condition has to be imposed on the effects, $\sum_{i=1}^{v} \beta_{i} = 0$. In particular, for this identifiability condition the parameter relating to the last level $v$ can be recovered from the other levels, $ \beta_{v}=-\sum_{i=1}^{v-1} \beta_{i}$. Hence,
\begin{equation}\label{eq:6}
\begin{split}
\tilde{Y}_{n}(i)&=\mu_{n}+ \mathbf{f}(i)^{\top}\boldsymbol{\beta}+\tilde{\epsilon}_{n},
\end{split}
\end{equation}
where the reduced parameter vector $\boldsymbol{\beta}=(\beta_1,\dots,\beta_{v-1})^\top$.
\par
Then for paired comparisons an observation of the effects $\beta_i - \beta_j$ of level $i$ compared to level $j$ can be characterized by the response
\begin{equation}\label{eq:7}
\begin{split}
Y_{n}(i,j)&= (\mathbf{f}(i)-\mathbf{f}(j))^{\top}\boldsymbol{\beta}+\varepsilon_n = \beta_i-\beta_j+\varepsilon_n,
\end{split}
\end{equation}
where $\textbf{f}(i,j) = \textbf{e}_i-\textbf{e}_j$, $\textbf{f}(i,v) = \textbf{e}_i+\textbf{1}_{v-1}$, $ \textbf{f}(v,i)=-\textbf{f}(i,v) $, for $i,j = 1,\dots,v-1$ and $\textbf{f}(v,v) =\mathbf{0}$. \par
Note that the information matrix $\textbf{M}((i,j))=\frac{2}{v-1}(\textbf{Id}_{v-1}+\textbf{1}_{v-1}\textbf{1}^{\top}_{v-1})$ for $i \neq j$ where $\mathbf{Id}_m$ denotes the $m$-dimensional identity matrix, while $\textbf{M}((i,i))=\mathbf{0}$ for $i = j$. Because pair ($i,i$) with identical levels results in zero information, it can be neglected. Hence, only pair ($i,j$) with different levels should be used and that, in particular, the design $\bar{\xi}$ which assigns equal weight $1/(v(v-1))$ to each of the pair $(i,j)$, $i \neq j$ is optimal with resulting information matrix
\begin{equation}\label{eq:8}
\begin{split}
\textbf{M}(\bar{\xi})=\frac{2}{v-1}(\textbf{Id}_{v-1}+\textbf{1}_{v-1}\textbf{1}^{\top}_{v-1}).
\end{split}
\end{equation}
The corresponding information matrix $\mathbf{M}(\bar{\xi})$ has an inverse of the form
\begin{equation}
\begin{split}
\mathbf{M}(\bar{\xi})^{-1} =\frac{v-1}{2}(\mathbf{Id}_{v-1}-\frac{1}{v}\mathbf{1}_{v-1}\mathbf{1}^{\top}_{v-1}).
\end{split}
\end{equation}
This design $\bar{\xi}$ will serve as a starting point for constructing optimal designs in situations with more than one attribute later on. In particular, this concept developed for the single-attribute will be extended later on for describing the main-effects as well as the two- and the three-attribute interactions in the case of a couple of, say, $K$ attributes $k=1,\dots,K$. \par

It is worthwhile mentioning that under the indifference assumption of equal choice probabilities \citep[see e.g.][]{grossmann2002advances} the designs considered in this paper carry over to the \citet{bradley1952rank} type choice experiments in which the probability of choosing $\textbf{i}$ from the pair $(\textbf{i},\textbf{j})$ is given by $\exp[\textbf{f}(\textbf{i})^\top\boldsymbol\beta]/(\exp[\textbf{f}(\textbf{i})^\top\boldsymbol\beta]+\exp[\textbf{f}(\textbf{j})^\top\boldsymbol\beta])$. Specifically, this assumption simplifies the information matrix of the multinomial logit model which becomes proportional to the information matrix in the linear paired comparison model \citep[see][]{grossmann2015handbook,singh2015optimal}.

\section{Second-order interactions model}
In applications one may be interested in the utility estimates of both the main effects and interactions between the levels of the attributes. For that setting optimal designs have been derived \citep{van1987optimal, grasshoff2003optimal} in a two-attribute interactions setup. This paper considers a three-attribute interactions model. Corresponding results for the particular case of binary attributes can be found \citep{nyarko2019optimal}.\par
Following \citet{nyarko2019optimal}, we first start with the situation of full profiles. In that case each alternative is represented by level combinations in which all attributes are involved. For such alternatives we denote by $\textbf{i}=(i_1,\dots,i_K)$ and $\textbf{j}=(j_1,\dots,j_K)$ the first alternative and the second alternative, respectively, which are both elements of the set $\mathcal{I}=\{1,\dots,v\}^{K}$ where $1$ and $v$ represent the first and last level of each $k$th attribute, $k=1,\dots,K$. Here $(\textbf{i},\textbf{j})$ is an ordered pair of alternatives $\textbf{i}$ and $\textbf{j}$ which is chosen from the design region $\mathcal{X}=\mathcal{I}\times\mathcal{I}$. 
Note that for each attribute $k$ the corresponding regression functions $\textbf{f}_k=\textbf{f}:\mathcal{X}\rightarrow\textbf{R}^{v-1}$ coincide with that of the one-way attribute introduced in Example $1$.\par

In the presence of main effects up to three-attribute interactions direct responses $\tilde{Y}_{n}$ at alternative $\textbf{i}=(i_1,\ldots,i_K)$ can be modeled as
\begin{align}\label{eq:full_direct}
\tilde{Y}_{n}(\textbf{i})&=\mu_n+\sum_{k=1}^{K}\textbf{f}(i_k)^\top \boldsymbol{\beta}_k+\sum_{k<\ell}(\textbf{f}(i_k)\otimes \textbf{f}(i_\ell))^\top\boldsymbol{\beta}_{k\ell} \nonumber\\
&\qquad~~+\sum_{k<\ell<m}(\textbf{f}(i_k)\otimes \textbf{f}(i_\ell)\otimes \textbf{f}(i_m))^\top\boldsymbol{\beta}_{k\ell m}+\tilde{\varepsilon}_{n},
\end{align}
for full profiles, where $\otimes$ denotes the Kronecker product of vectors or matrices, $\boldsymbol{\beta}_k=(\beta_{i_k}^{(k)})_{i_k=1,\dots,v-1}$ denotes the main effect of the $k$th attribute, $\boldsymbol{\beta}_{k\ell}=(\beta_{i_ki_{\ell}}^{(k\ell)})_{i_k=1,\dots,v-1,~i_{\ell}=1,\dots,v-1}$ is the two-attribute interaction of the $k$th and $\ell$th attribute, and $\boldsymbol{\beta}_{k\ell m}=(\beta_{i_ki_{\ell}i_m}^{(k\ell m)})_{i_k=1,\dots,v-1,~i_{\ell}=1,\dots,v-1,~i_{m}=1,\dots,v-1}$ is the three-attribute interaction of the $k$th, $\ell$th and $m$th attribute.
The vectors $(\boldsymbol{\beta}_k)_{1\leq k\leq K}$ of main effects, $(\boldsymbol{\beta}_{k\ell})_{1\leq k<\ell\leq K}$ of two-attribute interactions and $(\boldsymbol{\beta}_{k\ell m})_{1\leq k<\ell<m\leq K}$ of three-attribute interactions have dimensions $p_1=K(v-1)$, $p_2=(1/2)K(K-1)(v-1)^{2}$ and $p_3=(1/6)K(K-1)(K-2)(v-1)^{3}$, respectively.
Hence the parameter vector 
\begin{align}\label{eqtg:4.}
\boldsymbol{\beta}=((\boldsymbol{\beta}_k)_{k=1,\dots,K}^\top,(\boldsymbol{\beta}_{k\ell})_{k<\ell}^\top,(\boldsymbol{\beta}_{k\ell m})^\top_{k<\ell<m})^\top,
\end{align}
 has dimension $p=p_1+p_2+p_3=K(v-1)(1+(1/6)(K-1)(v-1)(3+(K-2)(v-1)))$. 
These vector of parameters sum up to the complete parameter vector $\boldsymbol{\beta}\in\textbf{R}^p$. Here the corresponding $p$-dimensional vector of regression functions $\textbf{f}:\mathcal{X}\rightarrow\textbf{R}^p$ is given by
\begin{align}\label{eq:10}
\textbf{f}(\textbf{i})=(\textbf{f}(i_1)^\top,\dots,\textbf{f}(i_K)^\top,\textbf{f}(i_1)^\top\otimes\textbf{f}(i_2)^\top,\dots,\textbf{f}(i_{K-1})\otimes\textbf{f}(i_K)^\top,  \nonumber\\
\qquad \qquad \textbf{f}(i_1)^\top\otimes\textbf{f}(i_2)^\top\otimes\textbf{f}(i_3)^\top,\dots,\textbf{f}(i_{K-2})^\top\otimes\textbf{f}(i_{K-1})^\top\otimes\textbf{f} (i_K)^\top)^\top. 
\end{align}
Also here in $\textbf{f}(\textbf{i})$, the first $K$ components $\textbf{f}(i_{1}),\dots,\textbf{f}(i_{K})$ are associated with the main effects and have $p_1=K(v-1)$ parameters, the second components $\textbf{f}(i_1)\otimes\textbf{f}(i_2),\dots,\textbf{f}(i_{K-1})$ $\otimes\textbf{f}(i_K)$ are associated with the two-attribute interactions and have $p_2=(1/2)K(K-1)(v-1)^{2}$ parameters, and the remaining components $\textbf{f}(i_1)\otimes\textbf{f}(i_2)\otimes\textbf{f}(i_3),\dots,\textbf{f}(i_{K-2})\otimes\textbf{f}(i_{K-1})\otimes\textbf{f}(i_K)$ are associated with the three-attribute interactions and have $p_3=(1/6)K(K-1)(K-2)(v-1)^{3}$ parameters.\par
As was already pointed out, because of the limited cognitive to process information, in applications a preference task including many attributes may enhance respondent decisions that do not reflect their actual preferences. To overcome these behaviors, only partial profiles are presented  within a single paired comparison. Specifically, in a partial profile every pair consists of alternatives which are described by a predefined number $S$ of attributes, where the same set of attributes is used throughout both alternatives within a pair but with potentially different levels, and the remaining $K-S$ attributes are not shown and remain thus unspecified. The number $S$ of attributes used in a partial profile is called the profile strength. For instance, if in an experimental situation for a total number of attributes $K=11$ each at two levels, only $S=4$ of the attributes are shown, while the remaining $K-S=7$ attributes are not shown or officially set to $0$, then the $S=4$ of the attributes that are shown in the pairs of alternatives is the profile strength. \citet{green1974design} pointed out that in choosing a profile strength, it has to be taken into account: (a) the number of attributes to vary in each set of alternatives that will be presented to respondents; (b) the number of alternatives to present to respondents for evaluation; and (c) the type of utility (or paired comparison) model to apply in representing the respondent's evaluations. According to \citet{schwabe2003optimal}, typically choosing up to a maximum number of four attributes as representative of the profile strength while there may be twenty or more attributes is enough to reduce cognitive burden as frequently encountered in practice \citep[see][for example]{grossmann2017partial}.
\par

For a partial profile a direct observation may be described by model (\ref{eq:full_direct}) when summation is taken only over those $S$ attributes contained in the describing subset.
This requires that the profile strength $S\geq3$ is needed to ensure identifiability of the three-attribute interactions. In what follows, we introduce an additional level $i_k=0$, which indicates that the corresponding $k$th attribute is not present in the partial profile. The corresponding regression functions are extended to $\mathbf{f}(0)=\mathbf{0}$. With this convention a direct observation can be described by (\ref{eq:full_direct}) even for a partial profile $\textbf{i}$ from the set 
\begin{equation}\label{eq:12}
\begin{split}
\mathcal{I}^{(S)}=&\{\textbf{i};\ i_{k}\in\{1,\dots,v\} \textrm{ for $S$ components and} 
\\ 
&\qquad  
i_{k}=0 \textrm{ for $K-S$ components}\},
\end{split}
\end{equation}
of alternatives with profile strength $S$.
\par

For observations in linear paired comparisons the resulting model is given by
\begin{align}\label{eq:11}
Y_{n}(\textbf{i},\textbf{j})&=\sum_{k=1}^{K}(\textbf{f}(i_k)-\textbf{f}(j_k))^\top \boldsymbol{\beta}_k+\sum_{k<\ell}((\textbf{f}(i_k)\otimes \textbf{f}(i_\ell))-(\textbf{f}(j_k)\otimes \textbf{f}(j_\ell)))^\top\boldsymbol{\beta}_{k\ell} \nonumber\\
&\qquad~~+\sum_{k<\ell<m}((\textbf{f}(i_k)\otimes \textbf{f}(i_\ell)\otimes \textbf{f}(i_m))  \nonumber\\
&\qquad\qquad\qquad~~-(\textbf{f}(j_k)\otimes \textbf{f}(j_\ell)\otimes \textbf{f}(j_m)))^\top\boldsymbol{\beta}_{k\ell m}+\varepsilon_{n}.
\end{align}
The design region in this case can be specified as

\begin{equation}\label{eq:12}
\begin{split}
\mathcal{X}^{(S)}&=\{(\textbf{i},\textbf{j});~i_{k}, j_{k}\in\{1,\dots,v\} \ \textrm{for $S$ components and} \\
&\qquad\qquad \ \  i_{k}= j_{k}=0~\textrm{for exactly $K-S$ components}\},
\end{split}
\end{equation}
for the set of partial profiles with profile strength $S$. Notice that for full profiles each pair in the design region consists of alternatives where all attributes are shown.

\section{Optimal designs}
In the present setting we consider optimal designs for the three-attribute interactions paired comparison model \eqref{eq:11} with corresponding regression functions $\textbf{f}(\textbf{i})$ given by \eqref{eq:10}. In what follows, we define $d$ as the comparison depth \citep[see][]{grasshoff2003optimal}, which describes the number of attributes in which the two alternatives presented differ, $d=1,\dots, S$ \citep[ see][p.~9,~for example]{eric2019optimal4444}.
\par

For this situation the design region $\mathcal{X}^{(S)}$ in \eqref{eq:12} can be partitioned into disjoint sets 
\begin{equation}\label{eq:13}
\mathcal{X}^{(S)}_{d}=\{(\textbf{i},\textbf{j})\in\mathcal{X}^{(S)};\ i_{k}\neq j_{k} \textrm{ for exactly $d$ components}\}.
\end{equation}
These sets constitute the orbits with respect to permutations of both the levels $i_k,j_k=1,\dots,v$ within the attributes as well as among attributes $k=1,\dots,K$, themselves. 
\par

Note that the $D$-criterion is invariant with respect to those permutations, which induce a linear reparameterization \citep[see][p.~17]{1996optimum}. 
As a result, it is sufficient to look for optimality in the class of invariant designs \citep[see][p.~282-284, for detailed discussion]{kiefer1959optimum}. 
\par

Denote by $\bar{\xi}_{d}$ the uniform approximate design which assigns equal weights to each individual distinct (pair) design points in $\mathcal{X}^{(S)}_{d}$, $\bar{\xi}_{d}(\textbf{i},\textbf{j})=\frac{1}{N_{d}}$ where $N_{d}={K \choose S}{S \choose d}v^S(v-1)^d$ is the number of the distinct design points. Notice that for the associated identical (pair) design points in $\mathcal{X}^{(S)}$ there is no information available. It can be demonstrated that in this case, the corresponding uniform design $\bar{\xi}_{d}$ on the set $\mathcal{X}^{(S)}_{d}$ has an information matrix having a block diagonal structure. To begin with, we first note that $\mathbf{M}=\frac{2}{v-1}(\mathbf{Id}_{v-1}+\mathbf{1}_{v-1}\mathbf{1}^{\top}_{v-1})$ is the information matrix of the one-way layout in \eqref{eq:8}. Here, $\mathbf{Id}_m$ is the identity matrix of order $m$ for every $m$. In the following, the two functions $h_1(d)$ and $h_2(d)$ are identical to the corresponding terms for the main-effects and two-attribute interactions models \citep[e.g., see][]{grasshoff2003optimal}. 
\newtheorem{lemma}{Lemma}
\begin{lemma}\label{lemma1}
Let $d\in\{0,\dots,S\}$. The uniform design $\bar{\xi}_{d}$ on the set $\mathcal{X}_{d}^{(S)}$ of comparison depth $d$ has diagonal information matrix
\begin{equation*}
\begin{split}
\mathbf{M}(\bar{\xi}_{d})=\begin{pmatrix} h_{1}(d)\mathbf{Id}_{p_1}\otimes\mathbf{M} & \mathbf{0}& \mathbf{0}\\
 \mathbf{0} & h_{2}(d)\mathbf{Id}_{p_2}\otimes\mathbf{M}\otimes\mathbf{M}&\mathbf{0}\\
 \mathbf{0} &\mathbf{0}&h_{3}(d)\mathbf{Id}_{p_3}\otimes\mathbf{M}\otimes\mathbf{M}\otimes\mathbf{M}\end{pmatrix}
 \end{split}  
\end{equation*} 
where
\begin{equation*}
\begin{split}
h_{1}(d)& =\frac{d}{K},~h_{2}(d) =\frac{d}{2vK(K-1)}(2Sv-2S-dv-v+2)~\mathrm{and}~\\
h_{3}(d)& =\frac{d}{4v^{2}K(K-1)(K-2)}(3S^2+3S^2v^2-6S^2v-3Sdv^2+3Sdv-6Sv^2 \\
&\qquad\qquad\qquad\qquad+15Sv-9S+d^2v^2+3dv^2-6dv +2v^2-6v+6). 
\end{split}
\end{equation*}
\end{lemma}
\begin{proof}
The quantities $h_1(d)$ and $h_2(d)$ are identical to the terms in \citet{grasshoff2003optimal}. For $h_3(d)$ we proceed by first noting that the auxiliary terms $\sum^{v}_{i=1}\textbf{f}(i)\textbf{f}(i)^{\top}=\frac{v-1}{2}\textbf{M}$ and $\sum_{i\neq j}\textbf{f}(i)\textbf{f}(j)^{\top}=-\frac{v-1}{2}\textbf{M}$. \par
For the three-attribute interactions, we consider attributes $k$, $\ell$ and $m$, say, and distinguish between pairs in which all three attributes are distinct,  pairs in which two of these attributes $k$ and $\ell$, say, have distinct levels in the alternatives while the same level is presented in both alternatives for the remaining attribute and, finally, pairs in which only one of the attributes, say, $k$ has distinct levels in the alternatives while the same level is presented in both alternatives for the two remaining attributes:\vspace*{-4mm}
\begin{align}\label{eq:4.32}
&\sum_{i_{k}\neq j_{k}}\sum_{i_{\ell}\neq j_{\ell}}\sum_{i_{m}\neq j_{m}}(\textbf{f}(i_{k})\otimes\textbf{f}(i_{\ell})\otimes\textbf{f}(i_{m})-\textbf{f}(j_{k})\otimes\textbf{f}(j_{\ell})\otimes\textbf{f}(j_{m})) \nonumber\\
&\qquad\qquad\qquad\cdot(\textbf{f}(i_{k})\otimes\textbf{f}(i_{\ell})\otimes\textbf{f}(i_{m})-\textbf{f}(j_{k})\otimes\textbf{f}(j_{\ell})\otimes\textbf{f}(j_{m}))^{\top} \nonumber\\
&=\sum^{v}_{i_{k}=1}\sum_{j_{k}\neq i_{k}}\sum^{v}_{i_{\ell}=1}\sum_{j_{\ell}\neq i_{\ell}}\sum^{v}_{i_{m}=1}\sum_{j_{m}\neq i_{m}}(\textbf{f}(i_{k})\textbf{f}(i_{k})^{\top}\otimes\textbf{f}(i_{\ell})\textbf{f}(i_{\ell})^{\top}\otimes\textbf{f}(i_{m})\textbf{f}(i_{m})^{\top}  \nonumber\\
&\qquad+\textbf{f}(j_{k})\textbf{f}(j_{k})^{\top}\otimes\textbf{f}(j_{\ell})\textbf{f}(j_{\ell})^{\top}\otimes\textbf{f}(j_{m})\textbf{f}(j_{m})^{\top} \nonumber\\
&\qquad-\textbf{f}(i_{k})\textbf{f}(j_{k})^{\top}\otimes\textbf{f}(i_{\ell})\textbf{f}(j_{\ell})^{\top}\otimes\textbf{f}(i_{m})\textbf{f}(j_{m})^{\top}   \nonumber\\
&\qquad-\textbf{f}(j_{k})\textbf{f}(i_{k})^{\top}\otimes\textbf{f}(j_{\ell})\textbf{f}(i_{\ell})^{\top}\otimes\textbf{f}(j_{m})\textbf{f}(i_{m})^{\top})  \nonumber\\
&=2(v-1)^{3}\sum^{v}_{i_{k}=1}\textbf{f}(i_{k})\textbf{f}(i_{k})^{\top}\otimes\sum^{v}_{i_{\ell}=1}\textbf{f}(i_{\ell})\textbf{f}(i_{\ell})^{\top}\otimes\sum^{v}_{i_{m}=1}\textbf{f}(i_{m})\textbf{f}(i_{m})^{\top} \nonumber\\
&\qquad-2\sum_{i_{k}\neq j_{k}}\textbf{f}(i_{k})\textbf{f}(j_{k})^{\top}\otimes\sum_{i_{\ell}\neq j_{\ell}}\textbf{f}(i_{\ell})\textbf{f}(j_{\ell})^{\top}\otimes\sum_{i_{m}\neq j_{m}}\textbf{f}(i_{m})\textbf{f}(j_{m})^{\top}   \nonumber\\
&=\frac{1}{4}v(v-1)^{3}(v^{2}-3v+3)\textbf{M}\otimes\textbf{M}\otimes\textbf{M},
\end{align}
also
\begin{align}\label{eq:4.33}
&\sum_{i_{k}\neq j_{k}}\sum_{i_{\ell}\neq j_{\ell}}\sum_{i_{m}= j_{m}}(\textbf{f}(i_{k})\otimes\textbf{f}(i_{\ell})\otimes\textbf{f}(i_{m})-\textbf{f}(j_{k})\otimes\textbf{f}(j_{\ell})\otimes\textbf{f}(j_{m}))   \nonumber\\
&\qquad\qquad\qquad\cdot(\textbf{f}(i_{k})\otimes\textbf{f}(i_{\ell})\otimes\textbf{f}(i_{m})-\textbf{f}(j_{k})\otimes\textbf{f}(j_{\ell})\otimes\textbf{f}(j_{m}))^{\top}\nonumber\\
&=\sum^{v}_{i_{k}=1}\sum_{j_{k}\neq i_{k}}\sum^{v}_{i_{\ell}=1}\sum_{j_{\ell}\neq i_{\ell}}\sum^{v}_{i_{m}=1}\sum_{j_{m}= i_{m}}(\textbf{f}(i_{k})\textbf{f}(i_{k})^{\top}\otimes\textbf{f}(i_{\ell})\textbf{f}(i_{\ell})^{\top}\otimes\textbf{f}(i_{m})\textbf{f}(i_{m})^{\top}  \nonumber\\
&\qquad+\textbf{f}(j_{k})\textbf{f}(j_{k})^{\top}\otimes\textbf{f}(j_{\ell})\textbf{f}(j_{\ell})^{\top}\otimes\textbf{f}(j_{m})\textbf{f}(j_{m})^{\top}   \nonumber\\
&\qquad-\textbf{f}(i_{k})\textbf{f}(j_{k})^{\top}\otimes\textbf{f}(i_{\ell})\textbf{f}(j_{\ell})^{\top}\otimes\textbf{f}(i_{m})\textbf{f}(j_{m})^{\top}  \nonumber\\
&\qquad-\textbf{f}(j_{k})\textbf{f}(i_{k})^{\top}\otimes\textbf{f}(j_{\ell})\textbf{f}(i_{\ell})^{\top}\otimes\textbf{f}(j_{m})\textbf{f}(i_{m})^{\top})  \nonumber\\
&=2(v-1)^{2}\sum^{v}_{i_{k}=1}\textbf{f}(i_{k})\textbf{f}(i_{k})^{\top}\otimes\sum^{v}_{i_{\ell}=1}\textbf{f}(i_{\ell})\textbf{f}(i_{\ell})^{\top}\otimes\sum^{v}_{i_{m}=1}\textbf{f}(i_{m})\textbf{f}(i_{m})^{\top}   \nonumber\\
&\qquad-2\sum_{i_{k}\neq j_{k}}\textbf{f}(i_{k})\textbf{f}(j_{k})^{\top}\otimes\sum_{i_{\ell}\neq j_{\ell}}\textbf{f}(i_{\ell})\textbf{f}(j_{\ell})^{\top}\otimes\sum_{i_{m}= j_{m}}\textbf{f}(i_{m})\textbf{f}(j_{m})^{\top}     \nonumber\\
&=\frac{1}{4}v(v-1)^{3}(v-2)\textbf{M}\otimes\textbf{M}\otimes\textbf{M},
\end{align}
and
\begin{align}\label{eq:4.34}
&\sum_{i_{k}\neq j_{k}}\sum_{i_{\ell}= j_{\ell}}\sum_{i_{m}= j_{m}}(\textbf{f}(i_{k})\otimes\textbf{f}(i_{\ell})\otimes\textbf{f}(i_{m})-\textbf{f}(j_{k})\otimes\textbf{f}(j_{\ell})\otimes\textbf{f}(j_{m})) \nonumber\\
&\qquad\qquad\qquad\cdot(\textbf{f}(i_{k})\otimes\textbf{f}(i_{\ell})\otimes\textbf{f}(i_{m})-\textbf{f}(j_{k})\otimes\textbf{f}(j_{\ell})\otimes\textbf{f}(j_{m}))^{\top}\nonumber\\
&=\sum^{v}_{i_{k}=1}\sum_{j_{k}\neq i_{k}}\sum^{v}_{i_{\ell}=1}\sum_{j_{\ell}= i_{\ell}}\sum^{v}_{i_{m}=1}\sum_{j_{m}= i_{m}}(\textbf{f}(i_{k})\textbf{f}(i_{k})^{\top}\otimes\textbf{f}(i_{\ell})\textbf{f}(i_{\ell})^{\top}\otimes\textbf{f}(i_{m})\textbf{f}(i_{m})^{\top}\nonumber\\
&\qquad+\textbf{f}(j_{k})\textbf{f}(j_{k})^{\top}\otimes\textbf{f}(j_{\ell})\textbf{f}(j_{\ell})^{\top}\otimes\textbf{f}(j_{m})\textbf{f}(j_{m})^{\top}\nonumber \\
&\qquad-\textbf{f}(i_{k})\textbf{f}(j_{k})^{\top}\otimes\textbf{f}(i_{\ell})\textbf{f}(j_{\ell})^{\top}\otimes\textbf{f}(i_{m})\textbf{f}(j_{m})^{\top}\nonumber\\
&\qquad-\textbf{f}(j_{k})\textbf{f}(i_{k})^{\top}\otimes\textbf{f}(j_{\ell})\textbf{f}(i_{\ell})^{\top}\otimes\textbf{f}(j_{m})\textbf{f}(i_{m})^{\top})\nonumber\\
&=2(v-1)\sum^{v}_{i_{k}=1}\textbf{f}(i_{k})\textbf{f}(i_{k})^{\top}\otimes\sum^{v}_{i_{\ell}=1}\textbf{f}(i_{\ell})\textbf{f}(i_{\ell})^{\top}\otimes\sum^{v}_{i_{m}=1}\textbf{f}(i_{m})\textbf{f}(i_{m})^{\top} \nonumber\\
&\qquad-2\sum_{i_{k}\neq j_{k}}\textbf{f}(i_{k})\textbf{f}(j_{k})^{\top}\otimes\sum_{i_{\ell}= j_{\ell}}\textbf{f}(i_{\ell})\textbf{f}(j_{\ell})^{\top}\otimes\sum_{i_{m}= j_{m}}\textbf{f}(i_{m})\textbf{f}(j_{m})^{\top}   \nonumber\\
&=\frac{1}{4}v(v-1)^{3}\textbf{M}\otimes\textbf{M}\otimes\textbf{M},  
\end{align}
respectively. \par

Now for the given attributes $k$, $\ell$ and $m$ the pairs with distinct levels in the three attributes occur $\left(\begin{smallmatrix}K-3 \\ S-3\end{smallmatrix}\right)\left(\begin{smallmatrix}S-3 \\ d-3\end{smallmatrix}\right)v^{S-3}(v-1)^{d-3}$ times in $\mathcal{X}^{(S)}_d$, while those which differ in two attributes occur $\left(\begin{smallmatrix}3 \\ 2\end{smallmatrix}\right)\left(\begin{smallmatrix}K-3 \\ S-3\end{smallmatrix}\right)\left(\begin{smallmatrix}S-3 \\ d-2\end{smallmatrix}\right)v^{S-3}(v-1)^{d-2}$ times in $\mathcal{X}^{(S)}_d$. Finally, those which differ only in one attribute occur $\left(\begin{smallmatrix}3 \\ 1\end{smallmatrix}\right)\left(\begin{smallmatrix}K-3 \\ S-3\end{smallmatrix}\right)\left(\begin{smallmatrix}S-3 \\ d-1\end{smallmatrix}\right)v^{S-3}(v-1)^{d-1}$
times in $\mathcal{X}^{(S)}_d$. As a consequence, the diagonal elements for the three-attribute interactions are given by

\begin{equation*}
\begin{split}
&\frac{1}{N_d}\bigg(\frac{1}{4}\left(\begin{matrix}K-3 \\ S-3\end{matrix}\right)\left(\begin{matrix}S-3 \\ d-3\end{matrix}\right)v^{S-2}(v-1)^{d}(v^{2}-3v+3)\textbf{M}\otimes\textbf{M}\otimes\textbf{M}\\
&\qquad+\frac{3}{4}\left(\begin{matrix}K-3 \\ S-3\end{matrix}\right)\left(\begin{matrix}S-3 \\ d-2\end{matrix}\right)v^{S-2}(v-1)^{d+1}(v-2)\textbf{M}\otimes\textbf{M}\otimes\textbf{M}\\
&\qquad+\frac{3}{4}\left(\begin{matrix}K-3 \\ S-3\end{matrix}\right)\left(\begin{matrix}S-3 \\ d-1\end{matrix}\right)v^{S-2}(v-1)^{d+2}(v-1)^{3}\textbf{M}\otimes\textbf{M}\otimes\textbf{M}\bigg)\\
&=\Big(\frac{d(d-1)(d-2)}{4v^{2}K(K-1)(K-2)}(v^{2}-3v+3)\\
&\qquad\qquad+\frac{3(S-d)d(d-1)}{4v^{2}K(K-1)(K-2)}(v-1)(v-2)   \\
&\qquad\qquad+\frac{3(S-d)(S-d-1)d}{4v^{2}K(K-1)(K-2)}(v-1)^{2}\Big)\textbf{M}\otimes\textbf{M}\otimes\textbf{M}\\
&=\frac{d}{4v^{2}K(K-1)(K-2)}(3S^2+3S^2v^2-6S^2v-3Sdv^2+3Sdv-6Sv^2  \\
&\qquad~~+15Sv-9S+d^2v^2+3dv^2-6dv +2v^2-6v+6) \textbf{M}\otimes\textbf{M}\otimes\textbf{M},
\end{split}
\end{equation*}
in the information matrix. It should be noted that the off-diagonal elements all vanish because the terms in the corresponding entries sum up to zero due to the effects-type coding.
\end{proof}

Note that for $d=0$ all pairs have identical attributes ($\textbf{i}=\textbf{j}$), $h_r(0)=0$ for $r=1,2,3$, and the information is zero. Hence, the comparison depth $d=0$ can be neglected. Further, invariant designs $\bar{\xi}$ can be written as a convex combination of uniform designs on the comparison depths $d$ with positive weights $w_{d}\geq 0$, $\sum^{S}_{d=1}w_{d}=1$. In this case the information matrix of the corresponding invariant design $\bar{\xi}$ also has the following diagonal structure: 
\begin{lemma}\label{lemma2}
Let $\bar{\xi}$ be an invariant design on $\mathcal{X}^{(S)}$. Then $\bar{\xi}$ has diagonal information matrix
\begin{multline*}
\begin{split}
\mathbf{M}(\bar{\xi})=\begin{pmatrix} h_{1}(\bar{\xi})\mathbf{Id}_{p_1}\otimes\mathbf{M} & \mathbf{0} & \mathbf{0} \\
 \mathbf{0} & h_{2}(\bar{\xi})\mathbf{Id}_{p_2}\otimes\mathbf{M}\otimes\mathbf{M}&\mathbf{0}\\
\mathbf{0} &\mathbf{0}&h_{3}(\bar{\xi})\mathbf{Id}_{p_3}\otimes\mathbf{M}\otimes\mathbf{M}\otimes\mathbf{M}\end{pmatrix}
 \end{split}
\end{multline*}
where $h_{r}(\bar{\xi})=\sum_{d=1}^Sw_dh_r(d)$, $r=1,2,3$.  
\end{lemma}
We now consider optimal designs for the main effects, the two-attribute interaction and the three-attribute interaction terms in the corresponding information matrix $\mathbf{M}(\bar{\xi}_{d})$ by maximizing the diagonal entries $h_{r}(d)$ for $r=1,2,3$. The resulting designs optimize every invariant design criterion \citep[see][]{nyarko2019optimal}. In the following, we mention that the optimal designs of Results~\ref{theorem1} and \ref{theorem2} can be recovered from the results of \citet{grasshoff2003optimal} for two-attribute interactions model. However, the result obtain for the three-attribute interactions is novel. Notice that $D$-optimality, which is discussed in \citep[][Section 4]{kiefer1959optimum}, refers to minimizing the determinant $\det \mathbf{M}(\bar{\xi}_d)^{-1}$ of the parameter estimates. In particular, $D$-optimal of the parameter vector of the main effects, two and three-attribute interaction effects as before is $p_1=K(v-1), p_2=(1/2)K(K-1)(v-1)^{2}$ and $p_3=(1/6)K(K-1)(K-2)(v-1)^{3}$, respectively.
 \newtheorem{result}{Result}
\begin{result}\label{theorem1}
The uniform design $\bar{\xi}_{S}$ on the largest possible comparison depth $S$ is $D$-optimal for the vector of main effects $(\boldsymbol{\beta}_{1}$ $\dots,$ $\boldsymbol{\beta}_{K})^{\top}$.
\end{result}
This means that for the main effects only those pairs of alternatives should be used which differ in all attributes subject to the profile strength.\par
For two-attribute interactions the number of the attributes subject to the profile strength $S$ with distinct levels for the uniform design $\bar{\xi}_{S}$ does not provide enough information. For example, the corresponding function $h_2(S)=0$ for the case $v=2$. As a consequence, only those pairs of alternatives should be used which differ in approximately half of the profile strength. In particular, one has to consider the following formula for calculating the optimal intermediate comparison depth $d^{\ast}_1=S-1-\begin{bmatrix}\frac{S-2}{v}\end{bmatrix}$ where $[q]$ denotes the integer part of the decimal expansion for $q$, satisfying $[q]\leq q< [q]+1$. 
\begin{result}\label{theorem2}
The uniform design $\bar{\xi}_{d^*_1}$ is $D$-optimal for the vector of two-attribute interaction effects $(\boldsymbol{\beta}_{k\ell})_{k<\ell}^{\top}$.
\end{result}
In the following Table \ref{Tab4.4h3PPROFILE}, we note that the values of $d^{\ast}_2$ were obtained by first calculating the values of $h_3(d)$ and determining the maximum. It is worthwhile mentioning that for very moderate values of $v$ ($v=2$, for example) the optimal comparison depth $d^{\ast}_2=S$ but this is not true for the case when $S=K=3$. Moreover, for sufficiently large values of $v$ ($v=20$, for example) the optimal comparison depth $d^*_2=S-2$.
\newtheorem{theorem}{Theorem}
\begin{theorem}\label{thrm22}
$\mathrm{(a)}$
For $S=3$ the uniform design $\bar{\xi}_{1}$ is $D$-optimal for the vector of three-attribute interaction effects $(\boldsymbol{\beta}_{k\ell m})_{k<\ell<m}^{\top}$.
\\
$\mathrm{(b)}$
For $S \geq 4$ the uniform design $\bar{\xi}_{d^*_2}$ is $D$-optimal for the vector of three-attribute interaction effects $(\boldsymbol{\beta}_{k\ell m})_{k<\ell<m}^{\top}$.
\end{theorem}

\begin{proof}
$\mathrm{(a)}$ Optimality is achieved when $h_3$ is maximized. For $S=3$ we get $h_3(1)=(v^2-2v+1)/4v^2p_3$ which establishes the result in this case.\par

$\mathrm{(b)}$ For $S\geq4$ note that $h_3$ is a polynomial of degree 3 in the comparison depth $d$ with positive leading coefficient. If we extend $h_3$ to a function defined on the real line, then it is point symmetric with respect to $((Sv-S-v+2)/v,h_3((Sv-S-v+2)/v))$ and attains its local minimum at $d_{3, \min}=(Sv-S-v+2)/v+\sqrt{9Sv+3v^2-9S-18v+18}/(3v)$. 
The result of $h_3(d_{3, \min})$ is equal to $c(3Sv^2-3Sv-3S-d_{3, \min}v^2-3v^2+6v)$ for some suitable positive constant $c$ depending on both $K$ and $S$. Inserting the solution for $d_{3, \min}$ into the last factor yields $3Sv^2-3Sv-3S-d_{3, \min}v^2-3v^2+6v=2Sv^2-2Sv-3S-v\sqrt{Sv+v^2/3-S-2v+2}-2v^2+4v>0$ for $S\geq 4$ and $v\geq 2$. Hence, $h_3(d_{3, \max})>0$ and, by symmetry, the value of $h_3$ at the local maximum $d_{3, \max}=(Sv-S-v+2)/v-\sqrt{9Sv+3v^2-9S-18v+18}/(3v)$ is equal to $h_3(d_{3, \max})<h_3(S)$ which proves that the global maximum of $h_3$ is attained at $d\leq S$ for $0\leq d\leq S$.
\end{proof}
\begin{table}[H]
\centering
\setlength\tabcolsep{0pt}
\caption{Values of the optimal comparison depths $d^*_2$ of the $D$-optimal uniform designs $\bar{\xi}_{d^{\ast}_2}$ for the three-attribute interactions in the case of full profiles $(S=K)$ and $v$-levels}\label{Tab4.4h3PPROFILE} 
\begin{tabular*}{\linewidth}{@{\extracolsep{\fill}}
    *{11}{D{.}{.}{3}}
                }
    \toprule
  & \multicolumn{9}{c}{$v$} \\
    \cmidrule(lr){2-11}
   K & 2       &    3     &  4 & 5  &  6 & 7    &    8& 9&10& 20                         \\
    \hline
3&1&1&  1 &1&1&1&1&1&1&1\\  
4&4&1&  2 &2&2&2&2&2&2&2\\   
5&5&2&  2 &3&3&3&3&3&3&3\\  
6&6&6& 3  &3&3&4&4&4&4&4\\ 
7&7&7& 7  &4&4&4&4&5&5&5\\ 
8&8&8& 8  &5&5&5&5&5&6&6\\    
9&9&9&9   &9&6&6&6&6&6&7\\    
10&10&10&10   &10&6&7&7&7&7&8\\  \bottomrule
\end{tabular*}
    \end{table}
    \par
The numerical results presented in Table \ref{Tab4.4h3PPROFILE} means that those pairs of alternatives should be used which differ in the comparison depth $d^*_2$ subject to the profile strength $S$. It is worth pointing out that by using a single comparison depth, it is possible to identify or estimate all the parameters of the corresponding three-attribute interactions model according to some criterion like the $D$-criterion. According to the Kiefer-Wolfowitz equivalence theorem \citep{kiefer1960equivalence}, a design $\xi^{\ast}$ is $D$-optimal if the variance function defined as 
$V((\textbf{i},\textbf{j}),\xi)=(\textbf{f}(\textbf{i})-\textbf{f}(\textbf{j}))^{\top}\textbf{M}(\xi)^{-1}(\textbf{f}(\textbf{i})-\textbf{f}(\textbf{j}))$, whenever $\textbf{M}(\xi)$ is non-singular is bounded by the number of model parameters defined as $p=p_1+p_2+p_3$, $V((\textbf{i},\textbf{j}),\xi^{\ast})\leq p$. This variance function plays an important role for the $D$-criterion. \par

For invariant designs $\bar{\xi}$ the value of the variance function evaluated at comparison depth $d$ may be denoted by $V(d,\bar{\xi})$, say, where $V(d,\bar{\xi})=V((\textbf{i},\textbf{j}),\bar{\xi})$ on $\mathcal{X}^{(S)}_{d}$. It can be shown that in this case the variance function $V(d,\bar{\xi})$ of the invariant design $\bar{\xi}$ has the following structure.
\begin{theorem}\label{thrm4}
For every invariant design $\bar{\xi}$ the variance function $V(d,\bar{\xi})$ is given by 
\begin{equation*}
\begin{split}
&\resizebox{1.0\hsize}{!}{$
V(d,\bar{\xi})=d(v-1)\Big(\frac{1}{h_{1}(\bar{\xi})}+\frac{v-1}{4vh_{2}(\bar{\xi})}(2Sv-2S-dv-v+2)+\frac{(v-1)^2}{24v^2h_{3}(\bar{\xi})}\lambda(d)\Big),$}
\end{split}
\end{equation*}
where
\begin{align*}
\begin{split}
&\lambda(d)=3S^2+3S^2v^2-6S^2v-3Sdv^2+3Sdv-6Sv^2+15Sv-9S+d^2v^2 \\
&\qquad\qquad+3dv^2-6dv +2v^2-6v+6.
\end{split}
\end{align*}
\end{theorem}
\begin{proof}
First we note that
\begin{equation*}
\begin{split}
\mathbf{M}(\bar{\xi})^{-1}=\begin{pmatrix} \frac{1}{h_{1}(\bar{\xi})}\textbf{Id}_{p_1}&\mathbf{0}&\mathbf{0}\\
 \mathbf{0} & \frac{1}{h_{2}(\bar{\xi})}\mathbf{Id}_{p_2}&\mathbf{0}\\
 \mathbf{0}&\mathbf{0}&\frac{1}{h_{3}(\bar{\xi})}\mathbf{Id}_{p_3}\end{pmatrix},
 \end{split}
\end{equation*}
for the inverse of the information matrix of the design $\bar{\xi}$. \par
Now, by Lemma $2$ of \citet{grasshoff2003optimal} it is sufficient to note that for the $k$-th main effects the variance function is given by
\begin{align}\label{eq:4.36}
&(\textbf{f}(i_k)-\textbf{f}(j_k))^{\top}\textbf{M}^{-1}(\textbf{f}(i_k)-\textbf{f}(j_k)) =v-1.
\end{align}
Further for the regression function associated with the two-attribute interactions of the attributes $k$ and $\ell$, say, we obtain 
\begin{align}\label{eq:4.37}
&(\textbf{f}(i_{k})\otimes\textbf{f}(i_{\ell})-\textbf{f}(j_{k})\otimes\textbf{f}(j_{\ell}))^{\top}\textbf{M}^{-1}\otimes\textbf{M}^{-1}(\textbf{f}(i_{k})\otimes\textbf{f}(i_{\ell})-\textbf{f}(j_{k})\otimes\textbf{f}(j_{\ell}))   \nonumber\\
&=\textbf{f}(i_{k})^{\top}\textbf{M}^{-1}\textbf{f}(i_{k}) \cdot \textbf{f}(i_{\ell})^{\top}\textbf{M}^{-1}\textbf{f}(i_{\ell})+ \textbf{f}(j_{k})^{\top}\textbf{M}^{-1}\textbf{f}(j_{k})\cdot \textbf{f}(j_{\ell})^{\top}\textbf{M}^{-1}\textbf{f}(j_{\ell})  \nonumber\\
&~~~-\textbf{f}(i_{k})^{\top}\textbf{M}^{-1}\textbf{f}(j_{k})\cdot \textbf{f}(i_{\ell})^{\top}\textbf{M}^{-1}\textbf{f}(j_{\ell})- \textbf{f}(j_{k})^{\top}\textbf{M}^{-1}\textbf{f}(i_{k}) \cdot \textbf{f}(j_{\ell})^{\top}\textbf{M}^{-1}\textbf{f}(i_{\ell})   \nonumber\\
&=\begin{dcases}\frac{(v-1)^{2}(v-2)}{2v}  & \text{for $i_{k}\neq j_{k},i_{\ell}\neq j_{\ell}$} \\
\frac{(v-1)^{3}}{2v}   & \text{for $i_{k}\neq j_{k},i_{\ell}= j_{\ell}$ or $i_{k}= j_{k},i_{\ell}\neq j_{\ell}$}.
\end{dcases}
\end{align} 
Accordingly, for the regression function associated with the interaction of the attributes $k$, $\ell$ and $m$, say, we obtain
\begin{align}\label{eq:4.38}
&(\textbf{f}(i_{k})\otimes\textbf{f}(i_{\ell})\otimes\textbf{f}(i_{m})-\textbf{f}(j_{k})\otimes\textbf{f}(j_{\ell})\otimes\textbf{f}(j_{m}))^{\top}\textbf{M}^{-1}\otimes\textbf{M}^{-1}\otimes\textbf{M}^{-1}  \nonumber \\
&\qquad\qquad\cdot(\textbf{f}(i_{k})\otimes\textbf{f}(i_{\ell})\otimes\textbf{f}(i_{m})-\textbf{f}(j_{k})\otimes\textbf{f}(j_{\ell})\otimes\textbf{f}(j_{m}))\nonumber\\
&\qquad=\textbf{f}(i_{k})^{\top}\textbf{M}^{-1}\textbf{f}(i_{k}) \cdot \textbf{f}(i_{\ell})^{\top}\textbf{M}^{-1}\textbf{f}(i_{\ell}) \cdot \textbf{f}(i_{m})^{\top}\textbf{M}^{-1}\textbf{f}(i_{m})  \nonumber\\
&\qquad\qquad+ \textbf{f}(j_{k})^{\top}\textbf{M}^{-1}\textbf{f}(j_{k})\cdot \textbf{f}(j_{\ell})^{\top}\textbf{M}^{-1}\textbf{f}(j_{\ell})\cdot \textbf{f}(j_{m})^{\top}\textbf{M}^{-1}\textbf{f}(j_{m}) \nonumber\\
&\qquad\qquad-\textbf{f}(i_{k})^{\top}\textbf{M}^{-1}\textbf{f}(j_{k})\cdot \textbf{f}(i_{\ell})^{\top}\textbf{M}^{-1}\textbf{f}(j_{\ell})\cdot \textbf{f}(i_{m})^{\top}\textbf{M}^{-1}\textbf{f}(j_{m})   \nonumber\\
&\qquad\qquad-\textbf{f}(j_{k})^{\top}\textbf{M}^{-1}\textbf{f}(i_{k}) \cdot \textbf{f}(j_{\ell})^{\top}\textbf{M}^{-1}\textbf{f}(i_{\ell})\cdot \textbf{f}(j_{m})^{\top}\textbf{M}^{-1}\textbf{f}(i_{m}) \nonumber\\
&\qquad=\begin{dcases}\frac{(v-1)^{3}(v^2-3v+3)}{4v^2}  & \text{for $i_{k}\neq j_{k},i_{\ell}\neq j_{\ell},i_{m}\neq j_{m}$} \\
\frac{(v-1)^{4}(v-2)}{4v^2}   & \text{for $i_{k}\neq j_{k},i_{\ell}\neq j_{\ell},i_{m}= j_{m}$} \\
\frac{(v-1)^{5}}{4v^2}   & \text{for $i_{k}\neq j_{k},i_{\ell}= j_{\ell},i_{m}= j_{m}$}.
\end{dcases}
\end{align} \par
Now for a pair of alternatives $(\textbf{i},\textbf{j})\in\mathcal{X}^{(S)}_d$ of comparison depth $d$: there are exactly $d$ attributes of the main effects for which $i_{k}$ and $j_{k}$ differ, there are $\frac{1}{2}d(d-1)$ two-attribute interaction terms for which $(i_{k}i_{\ell})$ and $(j_{k}j_{\ell})$ differ in all two attributes $k$ and $\ell$, there are $d(S-d)$ two-attribute interaction terms for which $(i_{k}i_{\ell})$ and $(j_{k}j_{\ell})$ differ in exactly one attribute $k$ or $\ell$, there are $\frac{1}{6}d(d-1)(d-2)$ three-attribute interaction terms for which $(i_{k}i_{\ell}i_m)$ and $(j_{k}j_{\ell}j_m)$ differ in all three attributes $k$, $\ell$ and $m$, there are $\frac{1}{2}(S-d)d(d-1)$ three-attribute interaction terms for which $(i_{k}i_{\ell}i_m)$ and $(j_{k}j_{\ell}j_m)$ differ in exactly two of the associated three attributes and finally, there are $\frac{1}{2}(S-d)(S-d-1)d$ three-attribute interaction terms for which $(i_{k}i_{\ell}i_m)$ and $(j_{k}j_{\ell}j_m)$ differ in exactly one of the associated three attributes. As a consequence, we obtain
\begin{equation*}
\begin{split}
V(d,\bar{\xi})&=(\textbf{f}(\textbf{i})-\textbf{f}(\textbf{j}))^{\top}\textbf{M}(\bar{\xi})^{-1}(\textbf{f}(\textbf{i})-\textbf{f}(\textbf{j}))\\
&=\frac{d(v-1)}{h_{1}(\bar{\xi})}+\frac{d(d-1)}{2}\frac{(v-1)^2(v-2)}{2vh_{2}(\bar{\xi})} +d(S-d)\frac{(v-1)^3}{2vh_{2}(\bar{\xi})}    \\
&\qquad+\frac{d(d-1)(d-2)}{6}\frac{(v-1)^3(v^2-3v+3)}{4v^2h_{3}(\bar{\xi})}  \\
&\qquad+\frac{(S-d)d(d-1)}{2}\frac{(v-1)^4(v-2)}{4v^2h_{3}(\bar{\xi})}\\
&\qquad+\frac{(S-d)(S-d-1)d}{2}\frac{(v-1)^5}{4v^2h_{3}(\bar{\xi})}\\
&=\frac{d(v-1)}{h_{1}(\bar{\xi})}+\frac{d(v-1)^2}{4vh_{2}(\bar{\xi})}\begin{pmatrix}(d-1)(v-2)+2(S-d)(v-1)\end{pmatrix}      \\
&\qquad+\frac{d(v-1)^3}{24v^2h_{3}(\xi)}\Big((d-1)(d-2)(v^2-3v+3)  \\
&\qquad\qquad\qquad\qquad+3(S-d)(d-1)(v-1)(v-2)  \\
&\qquad\qquad\qquad\qquad+3(S-d)(S-d-1)(v-1)^2\Big)\\
&=\frac{d(v-1)}{h_{1}(\bar{\xi})}+\frac{d(v-1)^2}{4vh_{2}(\bar{\xi})}(2Sv-2S-dv-v+2)     \\
&\qquad+\frac{d(v-1)^3}{24v^2h_{3}(\bar{\xi})}\big(3S^2v^2-6S^2v-6Sv^2+3S^2-3Sdv^2+3Sdv  \\
&\qquad\qquad\qquad\qquad+3dv^2+15Sv-9S+d^2v^2-6dv+2v^2-6v+6\big),\\
\end{split}
\end{equation*} 
for $(\textbf{i},\textbf{j})\in\mathcal{X}^{(S)}_d$ which proves the proposed formula. 
\end{proof}

In the case of a single comparison depth it can be shown that the corresponding invariant design $\bar{\xi}$ has the following structure.
\newtheorem{corollary}{Corollary}
\begin{corollary}\label{cor_thrm4}
For a uniform design $\bar{\xi}_{d^{\prime}}$ on a single comparison depth $d^{\prime}$ the variance function is given by  
\begin{equation*}
\begin{split}
V(d,\bar{\xi}_{d^{\prime}})=\frac{d}{d^{\prime}}\begin{pmatrix}p_{1}+p_{2}\frac{2Sv-2S-dv-v+2}{2Sv-2S-d^{\prime}v-v+2}+p_{3}\frac{\lambda(d)}{\lambda(d^{\prime})}\end{pmatrix}.
\end{split}
\end{equation*}
\end{corollary}
Here $d$ is the comparison depth of the design $\bar{\xi}_{d}$ defined in Lemma $1$. It is worth noting that if $d=d^{\prime}$, then by the Kiefer-Wolfowitz equivalence theorem \citep{kiefer1960equivalence} the corresponding variance function $V(d,\bar{\xi}_{d^{\prime}})$ will be equal to the total number of the model parameters $p=p_1+p_2+p_3$.

 \par
The following result gives an upper bound on the number of comparison depths required for a $D$-optimal design.
\begin{theorem}\label{theorem5}
The $D$-optimal design $\bar{\xi}^{\ast}$ for the three-attribute interactions model is supported on, at most, three different comparison depths $S$, $d^{\ast}$ and $d^{\ast}+1$.
\end{theorem}
\begin{proof}
According to a corollary of Kiefer-Wolfowitz equivalence theorem \citep[][p.~364]{kiefer1960equivalence} for the $D$-optimal design $\bar{\xi}^{\ast}$ the variance function $V(d,\bar{\xi}^{\ast})$ is equal to the number of parameters $p$ for all $d$. By Theorem~\ref{thrm4} the variance function is a cubic polynomial in the comparison depth $d$ with positive leading coefficient. The variance function $V(d,\bar{\xi}^{\ast})$ may thus be equal to $p$ for, at most, three different values $d_1<d_2<d_3$ of $d$, say. Now, by the Kiefer-Wolfowitz equivalence theorem \citep{kiefer1960equivalence} itself $V(d,\bar{\xi}^{\ast})\leq p$ for all $d=0,1,\dots,S$.\par 
We now make use of the analytic tool of the Kiefer-Wolfowitz equivalence theorem to establish the $D$-optimality of the design $\bar{\xi}^{\ast}$ by direct maximization of $\ln(\det(\mathbf{M}(w_{d^{\ast}}^{\ast}\bar{\xi^{\ast}}_{d^{\ast}}+(1-w_{d^{\ast}}^{\ast})\bar{\xi}_S)))$ \citep[e.g., see][]{kiefer1960equivalence}. \par

By Lemma \ref{lemma1} the entries of the information matrix $\textbf{M}(\bar{\xi}^{\ast},w_S^{\ast})$ are specified by 
\begin{equation*}
\begin{split}
h_{1}(\bar{\xi}^{\ast},w_S^{\ast})&=w_S^{\ast}h_{1}(S)+(1-w_S^{\ast})h_{1}(d^{\ast})\\
&=\frac{Sv-S+Sw_S^{\ast}+2w_S^{\ast}v-3w_S^{\ast}-2v+3}{Kv},
\end{split}
\end{equation*}
\begin{equation*}
\begin{split}
h_{2}(\bar{\xi}^{\ast},w_S^{\ast})&=w_S^{\ast}h_{2}(K)+(1-w_S^{\ast})h_{2}(d^{\ast})=\frac{\lambda_1}{2K(K-1)v^2},
\end{split}
\end{equation*}
where
\begin{equation*}
\begin{split}
&\lambda_1=S^2v^2-2S^2v-Sv^2-S^2w_S^{\ast}-Sw_S^{\ast}v+2w_S^{\ast}v^2+S^2+3Sv+2Sw_S^{\ast}  \\
&\qquad-2v^2-5w_S^{\ast}v-2S+5v+3w_S^{\ast}-3,
\end{split}
\end{equation*}
and
\begin{equation*}
\begin{split}
h_{3}(\bar{\xi}^{\ast},w_S^{\ast})&=w_S^{\ast}h_{3}(S)+(1-w_S^{\ast})h_{3}(d^{\ast}) =\frac{\lambda_2}{4K(K-1)(K-2)v^3},
\end{split}
\end{equation*}
\begin{equation*}
\begin{split}
&\lambda_2=S^3v^3-3S^3v^2-3S^2v^3+3S^3v+S^3w_S^{\ast}+9S^2v^2+2Sv^3-4Sw_S^{\ast}v^2-S^3\\
&\qquad-9S^2v-3S^2w_S^{\ast}-2Sv^2+9Sw_S^{\ast}v+6w_S^{\ast}v^2+3S^2-3Sv-3Sw_S^{\ast}-6v^2  \\
&\qquad-15w_S^{\ast}v+3S+15v+9w_S^{\ast}-9.
\end{split}
\end{equation*}
Now, since the determinant of the information matrix $\textbf{M}(\bar{\xi}^{\ast},w_S^{\ast})$ is proportional to $h_{1}(\bar{\xi}^{\ast},w_S^{\ast})^{p_{1}}h_{2}(\bar{\xi}^{\ast},w_S^{\ast})^{p_{2}}h_{3}(\bar{\xi}^{\ast},w_S^{\ast})^{p_{3}}$, we thus obtain
\begin{equation*}
\begin{split}
\ln\det(\textbf{M}(\bar{\xi}^{\ast},w_S^{\ast}))&=c+ K(v-1)\cdot\ln(Sv-S+Sw_S^{\ast}+2w_S^{\ast}v-3w_S^{\ast}-2v+3)\\
&+\frac{K(K-1)(v-1)^2\cdot\ln \lambda_1}{2}+\frac{K(K-1)(K-2)(v-1)^3\cdot\ln \lambda_2}{6},
\end{split}
\end{equation*}
where $c$ is a constant independent of the weight $w_S^{\ast}$. Taking derivatives with respect to $w_S^{\ast}$ we obtain
\begin{equation*}
\begin{split}
&\frac{\partial}{\partial w_S^{\ast}}\ln\det(\textbf{M}(\bar{\xi}^{\ast},w_S^{\ast}))=K(v-1)\cdot\begin{pmatrix}\frac{S+2v-3}{Sv-S+Sw_S^{\ast}+2w_S^{\ast}v-3w_S^{\ast}-2v+3}\end{pmatrix}\\ 
&\ +\frac{K(K-1)(v-1)^2}{2\lambda_1}\cdot\begin{pmatrix}-S^2-Sv+2S+2v^2-5v+3\end{pmatrix}  \\
&\ +\frac{K(K-1)(K-2)(v-1)^3}{6\lambda_2}\cdot\begin{pmatrix}S^3-4Sv^2-3S^2-3S+9Sv+6v^2-15v+9\end{pmatrix}
\end{split}
\end{equation*}
which has root $w_S^{\ast}=1-w^{\ast}_{d^{\ast}}$. This root gives a maximum for the determinant. The design $\bar{\xi}^{\ast}$ is thus $D$-optimal when we consider the particular case of the reduced design region $\mathcal{X}_{S}\cup\mathcal{X}_{d^{\ast}}$. \par

Again, by inserting the corresponding functions $h_{1}(\bar{\xi}^{\ast},w_S^{\ast}),h_{2}(\bar{\xi}^{\ast},w_S^{\ast})$ and $h_{3}(\bar{\xi}^{\ast},w_S^{\ast})$ into the representation of the variance function $V(S,d,\bar{\xi}^{\ast},w_S^{\ast})$ in Theorem \ref{thrm4}, we obtain
\begin{equation*}
 V(S,d^{\ast},\bar{\xi}^{\ast},w_S^{\ast})= V(S,d^{\ast},d^{\ast}+1,\bar{\xi}^{\ast},w_S^{\ast})\leq p. 
\end{equation*}
Hence, $S,d^{\ast}$ and $d^{\ast}+1$ are integer solutions for the maximum of the variance function which shows the D-optimality of the design in view of the equivalence theorem by \citet{kiefer1960equivalence}. 

\end{proof}
For Results~\ref{theorem1}, \ref{theorem2} and Theorem~\ref{thrm22} on parts of the parameter vector, the $D$-optimal design for the complete parameter vector may depend on the profile strength $S$ as can be seen by the numerical examples for the case of arbitrary levels, $v\geq2$ presented in Table \ref{tab4.7}. We note that for the case $v=2$, $S=K=3$ of full profiles and complete interactions the $D$-optimal design which is given explicitly in \citet{nyarko2019optimal} indicate that all three comparison depths are needed for $D$-optimality. 

\par
For $S\geq 4$ numerical computations indicate that at most two different comparison depths $S$ and $d^{\ast}$ may be required for $D$-optimality. Because it is generally difficult to specify an explicit formula for calculating $d^*$, the following Table \ref{tab4.7} shows the corresponding optimal designs with their optimal comparison depths $d^{\ast}$ in boldface and their corresponding weights $w^{\ast}_{d^{\ast}}$ for various choices of attributes $K$ between $4$ and $10$ and levels $v=2,\dots,8$. Entries of the form $(d^{\ast}, w_{d^*}^{\ast})$ indicate that invariant designs $\bar{\xi^*}=w_{d^*}^*\bar{\xi^*}_{d^*}+(1-w_{d^*}^*)\bar{\xi}_S$ have to be considered, while for single entries $d^{\ast}$ the optimal design $\bar{\xi}^*= \bar{\xi}^*_{d^*}$ has to be considered which is uniform on the optimal comparison depth $d^{\ast}$. In particular, for the case $S=K=4$ of full profiles where $v=2$ the corresponding invariant design depends on the optimal comparison depths $d^*=2$ and $S=4$ with optimal weights $w_{d^*}^{\ast}=0.857$ and $w_{d^*}^{\ast}=0.143$, respectively. The optimality of the designs in Table~\ref{tab4.7} has been checked numerically by virtue of the Kiefer-Wolfowitz equivalence theorem. The values of the normalized variance function $V(d,\bar{\xi}^{\ast})/p$ are recorded in Table~\ref{tab4}, where maximal values less than or equal to $1$ establish optimality.

\begin{table}[H]
\centering
\caption{Optimal designs with intermediate comparison depths $d^{\ast}$ in boldface and optimal weights $w^{\ast}_{d^{\ast}}$ of the form $(d^{\ast}$, $w^{\ast}_{d^{\ast}})$ for the case of full profiles $(S=K)$ and $v$-levels}\label{tab4.7} 
\resizebox{!}{.08\paperheight}{
\begin{tabular}{cccccccc}\toprule
    \multicolumn{6}{c}{$v$}&\\
    \cline{2-8}
    $K$       &    2     &  3 & 4  &  5  & 6    &    7&  8                                \\
    \hline
4&(\textbf{2},\ 0.857)&\textbf{2}&\textbf{2}&\textbf{2}&\textbf{2}&\textbf{2}&\textbf{2}\\
   
5&(\textbf{2},\ 0.833)&(\textbf{2},\ 0.667)&\textbf{3}&\textbf{3}&\textbf{3}&\textbf{3}&\textbf{3}\\

6&(\textbf{3},\ 0.732)&(\textbf{3},\ 0.789)&\textbf{3}&\textbf{4}&\textbf{4}&\textbf{4}&\textbf{4}\\

7&(\textbf{3},\ 0.697)&(\textbf{4},\ 0.322)&\textbf{4}&\textbf{4}&\textbf{4}&\textbf{5}&\textbf{5}\\

8&(\textbf{3},\ 0.644)&\textbf{4}&(\textbf{5}, \ 0.425)&\textbf{5}&\textbf{5}&\textbf{5}&\textbf{5}\\

9&(\textbf{4},\ 0.577)&\textbf{5}&\textbf{5}&\textbf{6}&\textbf{6}&\textbf{6}&\textbf{6}\\

10&(\textbf{4},\ 0.538)&\textbf{5}&\textbf{6}&\textbf{6}&\textbf{7}&\textbf{7}&\textbf{7}\\ \bottomrule
\end{tabular}}
\end{table}

 \begin{table}[H]
\centering
\setlength\tabcolsep{0pt}
\caption{Values of the variance function $V(d,\bar{\xi}^\ast)$ for $\bar{\xi}^{\ast}$ from Table~\ref{tab4.7} in the case of full profiles ($S=K$) and $v$-levels (boldface \textbf{1} corresponds to the optimal comparison depths $d^*$)}\label{tab4} 
\begin{tabular*}{\linewidth}{@{\extracolsep{\fill}}
    *{12}{D{.}{.}{4}}
                }
    \toprule
  & \multicolumn{9}{c}{$d$} \\
    \cmidrule(lr){3-12}
   K & v       &    1     &  2 & 3  &  4  & 5    &    6&  7    &8&9&10                       \\
    \hline
4&2&0.875&\textbf{1} &0.875&\textbf{1}&&&&&&\\   
&3&0.813&\textbf{1} &0.938&1&&&&&&\\   
&4&0.793&\textbf{1} &0.953&0.983&&&&&&\\   
&5&0.783&\textbf{1} &0.962&0.980&&&&&&\\   
&6&0.777&\textbf{1} &0.968&0.980&&&&&&\\   
&7&0.773&\textbf{1} &0.973&0.981&&&&&&\\   
&8&0.770&\textbf{1} &0.976&0.982&&&&&&\\ \hline
5&2&0.760&\textbf{1} &0.960&0.880&\textbf{1}&&&&&\\   
&3&0.723&\textbf{1} &1&0.954&\textbf{1}&&&&&\\   
&4&0.689&0.967 &\textbf{1}&0.952&0.987&&&&&\\   
&5&0.666&0.951 &\textbf{1}&0.961&0.981&&&&&\\ 
&6&0.653&0.941 &\textbf{1}&0.968&0.980&&&&&\\ 
&7&0.644&0.934 &\textbf{1}&0.972&0.981&&&&&\\ 
&8&0.638&0.929 &\textbf{1}&0.976&0.982&&&&&\\  \hline
6&2&0.701&0.983 &\textbf{1}&0.906&0.855&\textbf{1}&&&&\\   
&3&0.624&0.921&\textbf{1}&0.968&0.932&\textbf{1}&&&&\\ 
&4&0.591&0.895 &\textbf{1}&0.993&0.963&0.997&&&&\\  
&5&0.576&0.882 &0.997&\textbf{1}&0.972&0.992&&&&\\  
&6&0.560&0.865&0.987&\textbf{1}&0.976&0.989&&&&\\  
&7&0.550&0.854 &0.981&\textbf{1}&0.979&0.988&&&&\\  
&8&0.543&0.846&0.977&\textbf{1}&0.982&0.988&&&&\\  \hline
7&2&0.615&0.917 &\textbf{1}&0.956&0.879&0.863&\textbf{1}&&&\\   
&3&0.553&0.860 &0.988&\textbf{1}&0.963&0.941&\textbf{1}&&&\\   
&4&0.519&0.822 &0.965&\textbf{1}&0.981&0.962&0.997&&&\\   
&5&0.498&0.800 &0.952&\textbf{1}&0.992&0.974&0.993&&&\\   
&6&0.487&0.787 &0.944&\textbf{1}&0.999&0.983&0.995&&&\\   
&7&0.479&0.777 &0.937&0.997&\textbf{1}&0.985&0.994&&&\\   
&8&0.471&0.768 &0.929&0.994&\textbf{1}&0.987&0.993&&&\\ \hline
8&2&0.559&0.872 &\textbf{1}&1&0.945&0.884&0.884&\textbf{1}&&\\   
&3&0.490&0.792&0.948&\textbf{1}&0.990&0.958&0.948&1&&\\   
&4&0.462&0.759 &0.924&0.993&\textbf{1}&0.980&0.969&\textbf{1}&&\\   
&5&0.442&0.732 &0.902&0.981&\textbf{1}&0.988&0.977&0.995&&\\   
&6&0.429&0.716 &0.889&0.974&\textbf{1}&0.994&0.982&0.994&&\\   
&7&0.421&0.706&0.880&0.970&\textbf{1}&0.997&0.987&0.995&&\\   
&8&0.415&0.698 &0.874&0.960&\textbf{1}&1&0.991&0.996&&\\ \hline
9&2&0.504&0.811 &0.962&\textbf{1}&0.969&0.910&0.868&0.883&\textbf{1}&\\   
&3&0.437&0.726 &0.894&0.972&\textbf{1}&0.969&0.946&0.946&1&\\   
&4&0.414&0.696&0.872&0.965&\textbf{1}&0.994&0.977&0.971&1&\\   
&5&0.397&0.674 &0.853&0.953&0.995&\textbf{1}&0.989&0.981&1&\\   
&6&0.384&0.657 &0.836&0.940&0.989&\textbf{1}&0.992&0.985&0.996&\\   
&7&0.376&0.645 &0.825&0.932&0.985&\textbf{1}&0.995&0.988&0.995&\\   
&8&0.370&0.637 &0.817&0.927&0.982&\textbf{1}&0.997&0.990&0.996&\\ \hline
10&2&0.462&0.763 &0.932&\textbf{1}&0.997&0.956&0.905&0.874&0.896&\textbf{1}\\   
&3&0.395&0.669 &0.843&0.938&\textbf{1}&0.972&0.953&0.938&0.947&1\\   
&4&0.374&0.642 &0.822&0.929&0.981&\textbf{1}&0.987&0.974&0.972&1\\   
&5&0.359&0.622 &0.803&0.917&0.977&\textbf{1}&1&0.989&0.985&1\\   
&6&0.348&0.606 &0.786&0.903&0.968&0.996&\textbf{1}&0.993&0.988&0.998\\   
&7&0.340&0.594 &0.774&0.892&0.961&0.993&\textbf{1}&0.995&0.990&0.997\\   
&8&0.335&0.586&0.765&0.885&0.956&0.990&\textbf{1}&0.996&0.991&0.996 \\  \bottomrule
\end{tabular*}
\end{table}
 
 As was already pointed out, the corresponding designs (Table~\ref{tab4.7}) possess large number of comparisons. For the case of binary attributes ($v=2$) the designs in Table~\ref{tab4.7} when $K=4, 5$ and $6$, for instance, consist of $96, 320$ and $1280$ pairs, respectively. In this situation it is possible to construct a design for estimating main effects and two and three attribute interactions. For example, if $K=3$, $v=2$ and $d^\ast=1$, the paired comparison design consists of $24$ pairs presented in the following Table~\ref{tab4.7333}:
 \begin{table}[H]
\centering
\caption{Design for estimating main effects plus two plus three attribute interactions of $K=3$ two-level attributes}\label{tab4.7333} 
\begin{tabular}{ccc}\toprule
    $n$       &    Pair $(\textbf{i}_n, \textbf{j}_n)$     &    \\\hline
1&((1,1,1,1,1,1,1),\ (2,1,1,2,2,1,2))\\ 
2&((1,2,1,2,1,2,2),\ (2,2,1,1,2,2,1))\\
3&((1,1,1,1,1,1,1),\ (2,1,1,2,2,1,2))\\ 
4&((1,2,1,2,1,2,2),\ (2,2,1,1,2,2,1))\\
5&((1,1,2,1,2,2,2),\ (2,1,2,2,1,2,1))\\
6&((1,2,2,2,2,1,1),\ (2,2,2,1,1,1,2))\\
7&((1,1,2,1,2,2,2),\ (2,1,2,2,1,2,1))\\
8&((1,2,2,2,2,1,1),\ (2,2,2,1,1,1,2))\\

9& ((1,1,1,1,1,1,1),\ (1,2,1,2,1,2,2))\\ 
10&((1,1,2,1,2,2,2),\ (1,2,2,2,2,1,1))\\
11&((1,1,1,1,1,1,1),\ (1,2,1,2,1,2,2))\\ 
12&((1,1,2,1,2,2,2),\ (1,2,2,2,2,1,1))\\
13&((2,1,1,2,2,1,2),\ (2,2,1,1,2,2,1))\\
14&((2,1,2,2,1,2,1),\ (2,2,2,1,1,1,2))\\
15&((2,1,1,2,2,1,2),\ (2,2,1,1,2,2,1))\\
16&((2,1,2,2,1,2,1),\ (2,2,2,1,1,1,2))\\

17&((1,1,1,1,1,1,1),\ (1,1,2,1,2,2,2))\\ 
18&((2,1,1,2,2,1,2),\ (2,1,2,2,1,2,1))\\
19&((1,1,1,1,1,1,1),\ (1,1,2,1,2,2,2))\\ 
20&((2,1,1,2,2,1,2),\ (2,1,2,2,1,2,1))\\
21&((1,2,1,2,1,2,2),\ (1,2,2,2,2,1,1))\\
22&((2,2,1,1,2,2,1),\ (2,2,2,1,1,1,2))\\
23&((1,2,1,2,1,2,2),\ (1,2,2,2,2,1,1))\\
24&((2,2,1,1,2,2,1),\ (2,2,2,1,1,1,2))\\

\bottomrule
\end{tabular}
\end{table}
 For convenience in notation, in the final design, the effects-coded levels $1$ and $-1$ are indicated by the first and last levels of the attributes, respectively. It should be noted that the design contains repeated pairs, which is as a result of the large number of the observations. This design is used to assess university students satisfaction with online teaching and psychological pressure on learning during the COVID-19 pandemic in Section $5$.

\section{Application}\label{aaa}
In this section we consider a practical situation where up to three factor interaction is of interest. In particular, we employ the design presented in Table~\ref{tab4.7333} to assess university students satisfaction with online teaching and psychological pressure on learning during the COVID-19 pandemic. The JMP Pro (Version 16.0) statistical software was used to analyzed the responses of 150 students of the University of Ghana who were intercepted on the university campus in the month of August 2021 (Fig. \ref{fig:1}). These results were further classified according to levels of psychological pressure on learning (Fig. \ref{eq:2}, Fig. \ref{eq:3} and Fig.\ref{fig:4}). The Log-Worth and the corresponding P-values of the various factors (or attributes) are also reported in the order of importance. In most of the cases, the three factor interactions perform well and important. \vspace*{-4mm}
\begin{figure}[H]
\includegraphics[width=1.2\textwidth]{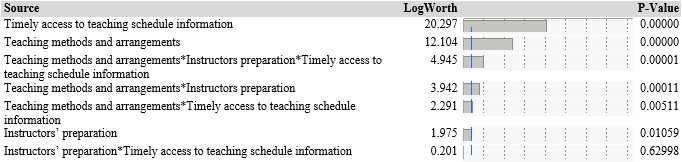}
\caption{Effect summary of university students satisfaction with online teaching during the COVID-19 pandemic}
\label{fig:1}   
\end{figure}
\vspace*{-6mm}
\begin{figure}[H]
\includegraphics[width=1.2\textwidth]{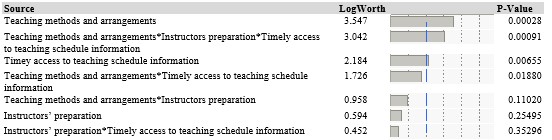}
\caption{Effect summary of students experiencing little psychological pressure during COVID-19 pandemic}
\label{fig:2}   
\end{figure}
\begin{figure}[H]
\includegraphics[width=1.2\textwidth]{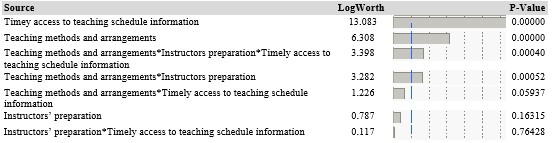}
\caption{Effect summary of students experiencing lots of psychological pressure during COVID-19 pandemic}
\label{fig:3}   
\end{figure}
\vspace*{-5mm}
\begin{figure}[H]
\includegraphics[width=1.2\textwidth]{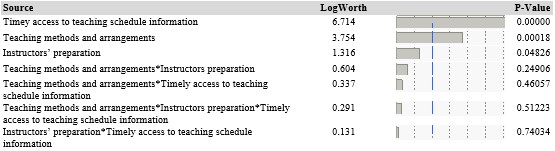}
\caption{Effect summary of students experiencing no psychological pressure during COVID-19 pandemic}
\label{fig:4}   
\end{figure}








\section{Discussion}
For paired comparisons where the alternatives are described by an analysis of variance model with main effects only optimal designs require that the paired alternatives show distinct levels in all attributes \citep[see][]{grasshoff2004optimal}. In a two-attribute interactions model pairs have to be used for an optimal design in which approximately $(v-1)/v$ of the attributes are distinct and $1/v$ of the attributes coincide \citep[see][]{grasshoff2003optimal}. Here it is shown that in a three-attribute interactions model one has to consider both types of pairs in which either all attributes have distinct levels or approximately $(v-1)/(v+1)$ of the attributes are distinct and $2/(v+1)$ of the attributes coincide to obtain a $D$-optimal design for the whole parameter vector \citep[see][arXiv]{nyarko2019optimalvlevels}. The resulting optimal designs for the particular situation of $v=2$ levels for each attribute have been obtained \citep{nyarko2019optimal}. Optimal designs may be concentrated on one, two or three different comparison depths depending on the number of levels and attributes. The so obtained designs can serve as a benchmark to judge the efficiency of any competing design as well as a starting point to construct fractions (or exact designs) which share the property of optimality and can be realized with a reasonable number of comparisons. A practical situation of importance is explored where the design enables identification of main effects and two and three attribute interactions \citep[see e.g.,][]{louviere2000stated, shah2015valuing}. 
 
\vspace{4mm} \noindent
\textbf{Acknowledgement}\vspace{1mm} \\
This work was partially supported by Grant - Doctoral Programmes in Germany, $2016/2017~(57214224)$ - of the German Academic Exchange Service (DAAD).
\bibliography{reference4th} 
\end{document}